\documentclass[runningheads]{llncs}
\usepackage[T1]{fontenc}
\usepackage{graphicx}
\usepackage{color}
\usepackage{amsmath,amssymb,amsfonts}
\usepackage{algorithm}
\usepackage{algpseudocode}
\usepackage{booktabs}
\usepackage{graphicx}
\usepackage{subcaption}
\usepackage{multirow}
\DeclareMathOperator{\dist}{dist}
\usepackage{svg}
\usepackage{hyperref}
\usepackage{esvect}
\usepackage{xcolor}
\usepackage{listings}
\usepackage{multirow}
\usepackage{float}
\usepackage{enumitem}
\usepackage{tikz}
\usetikzlibrary{arrows.meta}
\usetikzlibrary{positioning,shapes,arrows,calc}
\newcommand{\E}{\mathbb{E}}
\newcommand{\canon}{{canon}}
\usepackage{listings}
\usepackage{xcolor}
\usepackage{hyperref}
\begin{document}

\title{Expected Cost of Greedy Online Facility Assignment on Regular Polygons}

\titlerunning{Expected Cost Analysis of OFA on regular polygons}
\author{
    Md. Rawha Siddiqi Riad \and
    Md. Tanzeem Rahat \and
    Md. Manzurul Hasan
}
\institute{
    Department of Computer Science and Engineering,\\
    American International University-Bangladesh,\\
    Dhaka 1229, Bangladesh\\
    \email{22-46678-1@student.aiub.edu}, \email{tanzeem.rahat@aiub.edu} \email{manzurul@aiub.edu}
}
\authorrunning{Riad, M.R.S \and Rahat, M.T. \and Hasan, M.M.}
\maketitle

\begin{abstract}
We study a greedy online facility assignment process on a regular $n$-gon, where unit-capacity facilities occupy the vertices and customers arrive sequentially at uniformly random locations on polygon edges. 
Each arrival is irrevocably assigned to the nearest currently free facility under the shortest edge-walk metric, with uniform tie-breaking among equidistant choices. 
Our main theoretical result is an exact value-function characterization: for every occupancy state $S\subseteq V$, the expected remaining cost $V(S)$ satisfies a finite-horizon integral recurrence obtained by conditioning on the random arrival edge and position. 
To make this recurrence computationally effective, we exploit dihedral symmetry of the regular polygon and show that $V(S)$ is invariant under rotations and reflections, enabling canonicalization and symmetry-reduced dynamic programming. 
For small $n$, we evaluate the recurrence accurately using deterministic numerical integration over piecewise-linear distance regions,; for larger $n$, we estimate the expected total cost via direct Monte Carlo simulation of the online process and report $95\%$ confidence intervals. 
Our computations validate the recurrence (including a closed-form check for the square, $n=4$) and indicate that the total expected cost increases with $n$, while the per-customer expected travel distance grows gradually as remaining free vertices become farther on average.
\keywords{Online algorithms \and Facility assignment \and Expected cost \and Regular polygons \and Symmetry reduction \and Monte Carlo}
\end{abstract}

\section{Introduction}
\label{sec:intro}

Many modern service systems must allocate limited resources to requests that arrive over time without
advance knowledge. Examples include dispatching drivers to riders, assigning delivery jobs to vehicles,
and routing tasks to edge servers. A classical modeling approach is \emph{online facility location}, where an
algorithm may decide to open new facilities and connect demands online~\cite{meyerson2001online,fotakis2011online}.
In contrast, in many operational settings facility locations are fixed (e.g., depots, stations, servers) and
each facility has a hard capacity. This motivates the \emph{online facility assignment} (OFA) model, where each
arrival must be matched immediately to an available facility and the main source of difficulty is that early
assignments reduce future options~\cite{ahmed2020online}.

This paper studies OFA in a structured geometric environment: a regular $n$-gon with facilities at vertices.
Regular polygons provide a clean metric space where distances are defined by shortest walks along polygon edges,
yet the online coupling remains nontrivial: a greedy decision that is locally optimal can alter the geometry of
future availability. The polygon setting is also rich enough to capture symmetry and boundary effects that do not
appear on a line, and it has been the focus of recent worst-case (competitive) analyses~\cite{malik2025online}.
Our goal is complementary: we analyze the \emph{expected} performance of the greedy assignment rule under a fully
specified random arrival model.

Let $n\ge 3$ and consider a regular $n$-gon with vertex set $V=\{0,1,\dots,n-1\}$ and edge set
$E=\{(i,(i+1)\bmod n): i=0,\dots,n-1\}$. Each vertex hosts a unit-capacity facility (one customer per vertex).
Customers arrive sequentially. Each arrival independently selects an edge $e=(i,j)\in E$ uniformly at random,
then selects a position $t\in[0,1]$ uniformly along that edge, where $t=0$ corresponds to endpoint $i$ and $t=1$
to endpoint $j$.
The distance from an arrival $(e,t)$ to a vertex $v$ is the shortest \emph{edge-walk} distance:
\begin{align*}
\text{dist}(e, t, v) &= \min\big(t + \delta(i, v), (1 - t) + \delta(j, v)\big), \\
\delta(u, v) &= \min(|u - v|, n - |u - v|).
\end{align*}
Upon arrival, the customer must be assigned irrevocably to a currently free facility minimizing $\dist(e,t,v)$,;
ties are broken uniformly at random among minimizers. The process terminates after exactly $n$ arrivals when all
facilities are occupied. The performance measure is the \emph{total assignment cost}, defined as the sum of travel
distances of all customers. W; we study its expectation over arrival randomness and tie-breaking.

Our contributions are designed to support both theoretical understanding and reproducible computation:

\begin{itemize}
\item \emph{Exact expected-cost recurrence (value-function formulation).}
We derive an exact finite-horizon integral recurrence for the value function $V(S)$, the expected remaining cost
from any occupancy state $S\subseteq V$. The recurrence conditions on the random arrival edge and position and
expresses $V(S)$ as immediate cost plus expected future cost over the next occupied vertex.

\item \emph{Symmetry reduction via dihedral invariance.}
We prove that $V(S)$ is invariant under rotations and reflections of the regular $n$-gon (dihedral group action).
This yields a practical algorithmic improvement: we canonicalize each state to a representative under dihedral
symmetry and run a symmetry-reduced dynamic program, substantially decreasing the number of distinct states to
evaluate.

\item \emph{Two-regime computation with quantified uncertainty.}
For small $n$, we evaluate the integral recurrence accurately using deterministic numerical integration over
piecewise-linear distance regions, obtaining high-precision values without Monte Carlo noise.
For larger $n$, we switch to direct Monte Carlo simulation of the entire online process and report $95\%$ confidence
intervals, enabling scalable estimation.

\item \emph{Analytic validation for $n=4$.}
As a sanity check and a guide for reviewers, we provide a complete closed-form computation for the square
($n=4$), which matches the values produced by our numerical methods.
\end{itemize}

The rest of the paper is organized as follows:
Section~\ref{sec:related} reviews related work in online facility location and online matching.
Section~\ref{sec:model} formalizes the polygon model and distance metric.
Section~\ref{sec:recurrence} derives the expected-cost recurrence for the value function.
Section~\ref{sec:symmetry} proves dihedral invariance and describes state canonicalization.
Section~\ref{sec:algorithms} presents computational methods for small and large $n$.
Section~\ref{sec:exp} reports experimental results and confidence intervals, and
Section~\ref{sec:conc} concludes with directions for future work.

\section{Related Work}\label{sec:related}
Our setting lies between online facility problems and online matching, with a key shift in viewpoint:
most prior work emphasizes worst-case (competitive) guarantees, while we analyze \emph{expected} cost under an
explicit random-arrival model.

Online facility location allows opening facilities over time and has been studied extensively via competitive
analysis, starting with Meyerson~\cite{meyerson2001online},; Fotakis~\cite{fotakis2011online} surveys subsequent
variants and refinements. In contrast, our model has \emph{fixed} facility locations and strict unit capacities,
so the algorithm only assigns each arrival to an available facility.

Ahmed et al.~\cite{ahmed2020online} formalized OFA on graph metrics and analyzed greedy policies in the worst
case. Malik et al.~\cite{malik2025online} study OFA specifically on regular polygons, again focusing on competitive
bounds. Our work complements these by computing average-case performance on polygons under uniform random arrivals.

OFA can be viewed as an online matching problem between arrivals and fixed resources. The classical adversarial
model is captured by Karp--Vazirani--Vazirani~\cite{karp1990optimal}. Random-input models often admit stronger
performance, as explored by Goel and Mehta~\cite{goel2008online}. Weighted online matching highlights additional
difficulties in worst-case settings~\cite{kalyanasundaram1993online},; Mehta~\cite{mehta2013online} surveys these
connections to resource allocation and ad allocation.

To evaluate expectations involving continuous arrivals and tie-breaking, we use sampling-based estimation rooted
in the Monte Carlo method of Metropolis and Ulam~\cite{metropolis1949monte},; Caflisch~\cite{caflisch1998monte}
surveys modern Monte Carlo and quasi-Monte Carlo techniques.

\smallskip
\noindent We provide an exact expected-cost recurrence for greedy OFA on regular polygons and
make it computationally practical via symmetry reduction and regime-appropriate numerical methods.

\section{Model and Preliminaries}\label{sec:model}

\subsection{Polygon, vertices, edges, and arrivals}\label{sec:model_arrivals}
Fix an integer $n\ge 3$. We consider a regular $n$-gon whose vertices host unit-capacity facilities.
Formally, the vertex set is
\[
V=\{0,1,\dots,n-1\},
\]
and the (undirected) edge set is the cycle
\[
E=\{(i,(i+1)\bmod n): i=0,1,\dots,n-1\}.
\]
Each vertex $v\in V$ contains a facility that can serve exactly one customer.; Tthus after a vertex is used once,
it becomes unavailable.

Customers arrive sequentially. Each arrival independently:
(i) chooses an edge $e\in E$ uniformly at random (probability $1/n$), and then
(ii) chooses a position $t\in[0,1]$ uniformly on that edge.
Throughout, we use the following \emph{orientation convention}:

\begin{quote}
If $e=(i,j)$, then $t=0$ corresponds to endpoint $i$ and $t=1$ corresponds to endpoint $j$.
Hence the distance along the edge from the arrival point to $i$ is $t$, and to $j$ is $1-t$.
\end{quote}

\subsection{Edge-walk distance}\label{sec:model_distance}
Distances are measured along polygon edges (i.e., in the cycle graph metric).
For vertices $u,v\in V$, define the (edge-count) cycle distance
\[
\delta(u,v)=\min\big(|u-v|,\; n-|u-v|\big).
\]
For an arrival at position $t$ on edge $e=(i,j)$ and a facility at vertex $v$, the shortest \emph{edge-walk}
distance is
\begin{equation}\label{eq:dist_def}
\dist(e,t,v)=\min\Big(t+\delta(i,v),\; (1-t)+\delta(j,v)\Big).
\end{equation}
The two terms correspond to reaching $v$ by first moving to endpoint $i$ or $j$, respectively, and then walking
along the polygon edges.

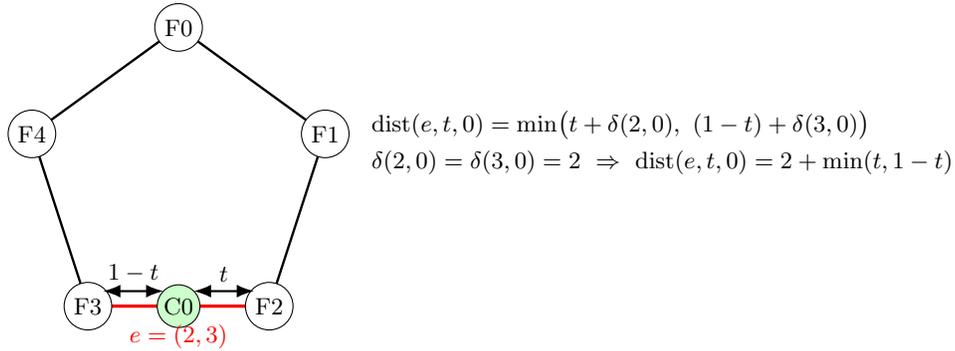
\begin{figure}[t]
\centering
\begin{tikzpicture}[
  >=Latex,
  facility/.style={circle, draw=black, fill=white, inner sep=1.6pt, minimum size=18pt, font=\small},
  customer/.style={circle, draw=black, fill=green!20, inner sep=1.2pt, minimum size=16pt, font=\small},
  edgeline/.style={line width=0.9pt},
  hiedge/.style={edgeline, red, line width=1.2pt},
  lab/.style={font=\small}
]
\def\R{2.05cm}
\def\tpos{0.50}

\coordinate (P0) at (90:\R);
\coordinate (P1) at (18:\R);
\coordinate (P2) at (-54:\R);
\coordinate (P3) at (-126:\R);
\coordinate (P4) at (162:\R);

\draw[edgeline] (P0)--(P1)--(P2)--(P3)--(P4)--cycle;

\draw[hiedge] (P2)--(P3);

\node[facility] (F0) at (P0) {F0};
\node[facility] (F1) at (P1) {F1};
\node[facility] (F2) at (P2) {F2};
\node[facility] (F3) at (P3) {F3};
\node[facility] (F4) at (P4) {F4};

\coordinate (C) at ($(P2)!\tpos!(P3)$);
\node[customer] (C0) at (C) {C0};

\node[lab, red] at ($(P2)!0.5!(P3) + (0,-0.4cm)$) {$e=(2,3)$};

\draw[<->, line width=0.8pt] ($(P2)+(-0.2cm,0.2cm)$) -- ($(C)+(0.2cm,0.2cm)$)
  node[midway, above, lab] {$t$};
\draw[<->, line width=0.8pt] ($(C)+(-0.2cm,0.2cm)$) -- ($(P3)+(0.2cm,0.2cm)$)
  node[midway, above, lab] {$1-t$};

\begin{scope}[xshift=3.45cm, yshift=0.2cm]
  \node[lab, anchor=west] at (-1,0.55)
  {$\dist(e,t,0)=\min\!\big(t+\delta(2,0),\ (1-t)+\delta(3,0)\big)$};
  \node[lab, anchor=west] at (-1,0.05) {$\delta(2,0)=\delta(3,0)=2\ \Rightarrow\ \dist(e,t,0)=2+\min(t,1-t)$};
\end{scope}
\end{tikzpicture}
\caption{Edge-walk distance illustration with the orientation convention: for $e=(2,3)$, the arrival is at distance $t$
from vertex $2$ and $1-t$ from vertex $3$. (Here F representes facilities and C represents customer.)}
\label{fig:dist_example}
\end{figure}

\subsection{Online assignment rule}\label{sec:model_rule}
At any time, let $S\subseteq V$ denote the set of already occupied vertices (facilities).
When a customer arrives at $(e,t)$, the algorithm assigns it to a free vertex
\[
v^\star \in \arg\min_{v\in V\setminus S}\dist(e,t,v).
\]
If the argmin contains multiple vertices (equal distance), the algorithm breaks ties uniformly at random among
the minimizers. The chosen vertex is then marked occupied, i.e., the state updates to $S\leftarrow S\cup\{v^\star\}$.
The process stops after exactly $n$ arrivals when $S=V$.
Our objective is to analyze the \emph{expected total assignment cost}, the sum of the travel distances incurred
over all arrivals.

\section{Expected-Cost Recurrence}\label{sec:recurrence}

\subsection{Value function}\label{sec:value_function}
A state is the set $S\subseteq V$ of already occupied facilities (vertices). Let $V(S)$ denote the
\emph{expected remaining total assignment cost} from state $S$ until all facilities are occupied, under the
greedy assignment rule described in Section~\ref{sec:model_rule}. In particular, the quantity of interest is
$V(\emptyset)$, the expected total cost starting from an empty polygon.

For a fixed state $S$, edge $e\in E$, and position $t\in[0,1]$, define the immediate assignment cost
\[
C_{S,e}(t) \;=\; \min_{v\in V\setminus S}\dist(e,t,v),
\]
where $\dist$ is the edge-walk distance from \eqref{eq:dist_def}. When multiple free vertices attain the minimum,
the greedy rule chooses uniformly among the minimizers.; Llet
$P_{v\mid S,e}(t)$ be the probability that the next occupied vertex is $v$ given $(S,e,t)$.
By definition, $P_{v\mid S,e}(t)=0$ for $v\in S$, and $\sum_{v\in V\setminus S}P_{v\mid S,e}(t)=1$.

\subsection{Main recurrence}\label{sec:main_recurrence}
The expected-cost recurrence follows from conditioning on the random arrival edge and location.

\begin{theorem}[Expected-cost recurrence]\label{thm:recurrence}
For any state $S\subseteq V$ with $|S|<n$, the value function satisfies
\begin{equation}\label{eq:Vrec}
V(S)
=
\frac{1}{n}\sum_{e\in E}\int_{0}^{1}
\left(
C_{S,e}(t) \;+\;
\sum_{v\in V\setminus S} P_{v\mid S,e}(t)\, V(S\cup\{v\})
\right)\,dt,
\end{equation}
with boundary condition $V(V)=0$.
\end{theorem}

\begin{proof}
Let $(E,T)$ be the random arrival, where $E$ is uniform over the $n$ edges and $T\sim\mathrm{Unif}[0,1]$ is
uniform along the chosen edge. Conditioning on $(E,T)$ and applying the law of total expectation yields
\begin{equation}\label{eq:totalexp}
\begin{split}
V(S) &= \mathbb{E}[\text{total remaining cost} \mid S] \\
     &= \frac{1}{n} \sum_{e \in E} \int_0^1 \mathbb{E}[\text{total remaining cost} \mid S, E=e, T=t] \, dt.
\end{split}
\end{equation}
Fix $S,e,t$. Under the greedy rule, the first arrival from this state incurs the deterministic
\emph{immediate cost} $C_{S,e}(t)$, and the process transitions to a new state
$S^{+}=S\cup\{v\}$ where $v$ is the (random) chosen minimizer vertex. Therefore,
\[
\E[\text{total remaining cost}\mid S,e,t]
=
C_{S,e}(t) \;+\; \E\big[V(S^{+}) \mid S,e,t\big].
\]
Taking expectation over tie-breaking outcomes,
\begin{align*}
\mathbb{E}\big[V(S^{+}) \mid S, e, t\big] &= \sum_{v \in V \setminus S} \Pr(S^{+} = S \cup \{v\} \mid S, e, t) \, V(S \cup \{v\}) \\
&= \sum_{v \in V \setminus S} P_{v \mid S, e}(t) \, V(S \cup \{v\}).
\end{align*}
Substituting into \eqref{eq:totalexp} gives \eqref{eq:Vrec}. Finally, if $S=V$ then no facility remains free
and the process has terminated, so $V(V)=0$.
\end{proof}
\qed

\subsection{Orientation-sensitive special cases}\label{sec:special_cases}
The recurrence in \eqref{eq:Vrec} is sensitive to the orientation convention for $t$
(Section~\ref{sec:model_arrivals}). In particular, when an endpoint of the arrival edge is occupied,
the immediate cost depends on which endpoint remains free.

\paragraph{Exactly one endpoint free.}
Let $e=(i,j)$.
If $i\notin S$ and $j\in S$ (only the start endpoint is free), then the greedy rule must choose $i$ for any
$t$, and
\[
C_{S,e}(t)=\dist(e,t,i)=t,\qquad
P_{i\mid S,e}(t)=1.
\]
If instead $i\in S$ and $j\notin S$ (only the end endpoint is free), then the greedy rule must choose $j$, and
\[
C_{S,e}(t)=\dist(e,t,j)=1-t,\qquad
P_{j\mid S,e}(t)=1.
\]
These two cases are equivalent up to swapping the orientation of $e$, but stating them explicitly avoids
ambiguity in implementations and figures.

\section{Structural Insight: Symmetry Reduction}\label{sec:symmetry}
A direct dynamic program over occupancy states has $2^n$ subsets and is quickly dominated by state-space size.
On a regular polygon, however, many states are equivalent up to symmetry: rotating or reflecting the polygon
does not change the distribution of future arrivals, nor the distances used by the greedy rule.
This section formalizes that invariance and shows how it yields an effective reduction in the number of distinct
DP states that must be evaluated.

\subsection{Dihedral invariance}\label{sec:dihedral}
Let $D_n$ denote the dihedral group of the regular $n$-gon, consisting of the $n$ rotations and $n$ reflections.
Each symmetry $g\in D_n$ acts as a relabeling of vertices $V=\{0,\dots,n-1\}$, and thus acts naturally on states
$S\subseteq V$ via $g(S)=\{g(v): v\in S\}$.

\begin{lemma}[Dihedral invariance of the value function]\label{lem:dihedral_invariance}
For every state $S\subseteq V$ and every dihedral symmetry $g\in D_n$,
\[
V(S) \;=\; V\big(g(S)\big).
\]
\end{lemma}

\begin{proof}
Consider the online process starting from state $S$. Apply $g$ as a relabeling of the polygon: vertices, edges,
and states are mapped to their images under $g$. Because the arrival model chooses an edge uniformly at random
and a position uniformly along that edge, the distribution of arrivals is invariant under dihedral symmetries:
the image of a uniformly random edge is again a uniformly random edge, and positions $t\in[0,1]$ remain uniform.

Moreover, the edge-walk distance \eqref{eq:dist_def} is preserved under this relabeling: for any edge
$e=(i,j)$, position $t$, and vertex $v$,
\[
\dist(e,t,v) \;=\; \dist\big(g(e),t,g(v)\big),
\]
since $g$ preserves adjacency on the cycle and thus preserves the cycle distance $\delta(\cdot,\cdot)$.
Therefore, the greedy decision rule (minimization of distance with uniform tie-breaking among minimizers)
commutes with $g$: the selected vertex in the transformed instance is exactly the image under $g$ of the selected
vertex in the original instance.

Consequently, the entire stochastic process of assignments and costs starting from $S$ has the same distribution
as the transformed process starting from $g(S)$, implying equality of expected total remaining costs:
$V(S)=V(g(S))$.
\end{proof}
\qed

\subsection{Canonicalization and symmetry-reduced memoization}\label{sec:canon}
Lemma~\ref{lem:dihedral_invariance} implies that the value function depends only on the orbit of $S$ under $D_n$,
not on the specific labeling of occupied vertices. We exploit this by storing DP values only for a canonical
representative of each orbit.

\paragraph{Bitmask representation.}
Represent a state $S$ as an $n$-bit mask $m(S)\in\{0,1\}^n$, where bit $i$ is $1$ iff $i\in S$.
Each symmetry $g\in D_n$ induces a permutation of bit positions, producing an image mask $m(g(S))$.

\paragraph{Canonical representative.}
Define
\[
\canon(S) \;=\; \arg\min_{g\in D_n} \; m\big(g(S)\big),
\]
where the minimum is taken in the integer (or lexicographic) order of bitmasks. Equivalently, interpret
$m(\cdot)$ as an integer and take the minimum value.
Thus, $\canon(S)$ is the orbit-minimal state under rotations and reflections.

\paragraph{Using canon($\cdot$) in DP.}
Whenever the DP needs $V(S)$, we replace $S$ by its canonical representative:
\[
V(S) \;=\; V\big(\canon(S)\big),
\]
and memoize only $V(\canon(S))$. This ensures that symmetric states are computed once and reused.
In practice, for moderate $n$ the number of distinct orbits is often close to a $2n$ factor smaller than $2^n$,
since a typical subset has approximately $2n$ distinct dihedral images (fewer only for highly symmetric subsets).

\paragraph{Practical impact.}
This symmetry reduction is crucial for making the recurrence in Theorem~\ref{thm:recurrence} computationally
usable. While worst-case complexity remains exponential, canonicalization reduces both runtime and memory by
eliminating redundant evaluations across symmetric occupancy patterns.

\section{Algorithms}\label{sec:algorithms}
The recurrence of Theorem~\ref{thm:recurrence} is exact but involves (i) exponentially many occupancy states and
(ii) a continuous integral over arrival positions. To obtain reliable numerical results across polygon sizes,
we use two complementary regimes: a symmetry-reduced dynamic program with deterministic integration for small
$n$, and direct Monte Carlo simulation of the full online process for larger $n$.

\subsection{Small $n$: symmetry-reduced DP with deterministic integration}\label{sec:alg_small}
For small to moderate $n$, we compute $V(\emptyset)$ by dynamic programming over states, exploiting the symmetry
reduction of Section~\ref{sec:canon}. The remaining challenge is evaluating, for each state $S$ and each edge $e$,
the integral over $t\in[0,1]$ in \eqref{eq:Vrec}.

\paragraph{Piecewise structure in $t$.}
Fix a state $S$ and an edge $e=(i,j)$. For each free vertex $v\in V\setminus S$, the function
\[
t \mapsto \dist(e,t,v)=\min\big(t+\delta(i,v),\ (1-t)+\delta(j,v)\big)
\]
is piecewise-linear in $t$ (a minimum of two affine functions). The immediate cost
$C_{S,e}(t)=\min_{v\in V\setminus S}\dist(e,t,v)$ is therefore also piecewise-linear, with breakpoints at values of
$t$ where two candidate distances coincide or where the minimizing path to a fixed $v$ switches between the
$i$-route and $j$-route. Consequently, the integrand in \eqref{eq:Vrec} is piecewise-linear except at finitely many
breakpoints, and tie points have measure zero.

\paragraph{Deterministic integration strategy.}
For each pair $(S,e)$ we:
\begin{enumerate}
\item compute a set of breakpoints $\mathcal{B}_{S,e}\subset[0,1]$ containing all $t$ where the identity of the
nearest free facility can change (including path-switch points for each free vertex and pairwise equality points
between candidate vertices),;
\item sort $\mathcal{B}_{S,e}$ and split $[0,1]$ into intervals on which the nearest-facility set is constant,;
\item on each interval, evaluate the integral using a fixed Gaussian quadrature rule (e.g., 4--8 points) or,
since the integrand is affine on that interval, integrate it exactly.
\end{enumerate}
On an interval where the nearest free vertex is unique, we have $P_{v\mid S,e}(t)\in\{0,1\}$ and the
future-cost term reduces to $V(S\cup\{v\})$. When ties occur only at isolated $t$ values, they do not affect the
integral.; Nnevertheless, for numerical stability we treat near-ties with a small tolerance.

\paragraph{DP ordering and memoization.}
We compute $V(S)$ in reverse order by $|S|$ (from $|S|=n$ down to $0$), storing values only for canonical states
$\canon(S)$. When evaluating the future term $V(S\cup\{v\})$, we immediately canonicalize the successor state.
This ensures that every distinct orbit under dihedral symmetry is computed once.

\begin{algorithm}[t]
\caption{Symmetry-reduced DP with deterministic integration (small $n$)}
\label{alg:small_dp}
\begin{algorithmic}[1]
\State \textbf{Input:} polygon size $n$
\State \textbf{Output:} $V(\emptyset)$
\State Initialize memo table $M[\cdot]$ for canonical states; set $M[\canon(V)]\gets 0$
\For{$k \gets n-1$ \textbf{down to} $0$}
  \For{\textbf{each} canonical state $S$ with $|S|=k$}
    \State $acc \gets 0$
    \For{\textbf{each} edge $e\in E$}
      \State compute breakpoints $\mathcal{B}_{S,e}$ and partition $[0,1]$ into intervals
      \State deterministically integrate the recurrence integrand over $t\in[0,1]$
      \State $acc \gets acc + \text{(edge integral)}$
    \EndFor
    \State $M[S] \gets acc/n$
  \EndFor
\EndFor
\State \Return $M[\canon(\emptyset)]$
\end{algorithmic}
\end{algorithm}

\subsection{Large $n$: direct Monte Carlo simulation of the full online process}\label{sec:alg_large}
For larger $n$, state-space DP (even with symmetry reduction) becomes expensive. Instead we estimate the expected
total cost by simulating the online assignment process directly.

\paragraph{One Monte Carlo run.}
A single run generates $n$ arrivals. Each arrival samples an edge uniformly from $E$, samples $t\sim\mathrm{Unif}[0,1]$,
computes $\dist(e,t,v)$ for all currently free vertices $v$, assigns the arrival to the nearest free vertex with
uniform tie-breaking, and accumulates the realized distance. The output of a run is a random total cost $X$.

\paragraph{Estimator and confidence intervals.}
Repeat the run independently for $R$ trials to obtain $X_1,\dots,X_R$.
We estimate the expected total cost by the sample mean
\[
\widehat{\mu}=\frac{1}{R}\sum_{r=1}^R X_r,
\]
with sample standard deviation $\widehat{\sigma}$. A $95\%$ confidence interval is reported as
\[
\widehat{\mu} \pm 1.96\cdot \frac{\widehat{\sigma}}{\sqrt{R}},
\]
which is appropriate for large $R$ by the central limit theorem. We also report per-customer cost
$\widehat{\mu}/n$ with the corresponding scaled interval.

\subsection{Complexity summary}\label{sec:complexity}
\paragraph{Small-$n$ DP.}
In the worst case the DP enumerates $2^n$ states, but symmetry reduction stores only one representative per
dihedral orbit, yielding a practical reduction close to a factor $2n$ for typical states.
For each stored state, the recurrence sums over $n$ edges and integrates over $t\in[0,1]$.
With breakpoint partitioning, the integrand is evaluated on a moderate number of intervals.; Uusing fixed-point
quadrature (or exact linear integration) adds a manageable constant factor. Overall, the approach remains
exponential in $n$ but is reliable and high-precision for small to moderate polygons.

\paragraph{Large-$n$ simulation.}
A naive simulation computes distances to all free vertices at each of the $n$ arrivals, giving
$O(n^2)$ time per run and $O(n)$ space. Over $R$ runs, this yields $O(Rn^2)$ time.
With additional data structures exploiting the cycle geometry, one can reduce the per-arrival search cost,
but the naive implementation is often sufficient for moderate $n$ and large $R$.

\section{Worked Validation Example: the Square ($n=4$)}\label{sec:square}
We provide a compact closed-form computation for $n=4$ to validate the recurrence in
Theorem~\ref{thm:recurrence} and to sanity-check our implementations. Label vertices
$V=\{0,1,2,3\}$ clockwise and edges $E=\{(0,1),(1,2),(2,3),(3,0)\}$. Let $C_k$ be the
(immediate) distance paid by the $k$-th arrival,; the total expected cost is
$V(\emptyset)=\E[C_1+C_2+C_3+C_4]$.

\paragraph{First arrival.}
With all facilities free, on any chosen edge the nearest endpoint is at distance $\min(t,1-t)$, hence
\[
\E[C_1]=\int_0^1 \min(t,1-t)\,dt
=2\int_0^{1/2} t\,dt=\frac14.
\]

\paragraph{Second arrival.}
By symmetry assume vertex $0$ was occupied first, so the free set is $\{1,2,3\}$.
Averaging over the four edges:
on $(0,1)$ the nearest free endpoint is $1$ with cost $1-t$; on $(3,0)$ it is $3$ with cost $t$; and on
$(1,2)$ and $(2,3)$ both endpoints are free giving cost $\min(t,1-t)$. Therefore,
\[
\E[C_2]=\frac14\Big(\int_0^1 (1-t)\,dt+\int_0^1 \min(t,1-t)\,dt+\int_0^1 \min(t,1-t)\,dt+\int_0^1 t\,dt\Big)
=\frac38.
\]

\paragraph{Third arrival.}
Given the first occupied vertex is $0$, the second occupied vertex is $1,2,3$ with probabilities
$3/8,\,1/4,\,3/8$, respectively, obtained by measuring which free vertex is closest over $t\in[0,1]$ on each edge.
Conditioning on the resulting two-occupied patterns yields
$\E[C_3]=19/32$.

\paragraph{Fourth arrival.}
With three vertices occupied, only one facility remains free (say vertex $0$ by symmetry). The edge-averaged
distance from a uniform arrival to vertex $0$ equals $1$, hence $\E[C_4]=1$.

\paragraph{Total.}
Summing,
\[
V(\emptyset)=\E[C_1]+\E[C_2]+\E[C_3]+\E[C_4]
=\frac14+\frac38+\frac{19}{32}+1=\frac{71}{32}=2.21875.
\]
This closed-form value matches the output of our numerical methods, validating the recurrence and the distance
convention.

\section{Experiments}\label{sec:exp}
\subsection{Setup}\label{sec:exp_setup}
We implemented both regimes from Section~\ref{sec:algorithms} in \texttt{C++} (double precision).
For symmetry reduction we canonicalize each occupancy bitmask under the $2n$ dihedral actions
(rotations and reflections) as described in Section~\ref{sec:canon}.
For the \emph{small-$n$} regime, the recurrence integrals are evaluated by deterministic numerical integration on
intervals where the nearest free facility is constant (with a conservative floating-point tolerance for
near-ties).
For the \emph{large-$n$} regime, we run $R$ independent simulations of the full online process.; Iin each run we
generate exactly $n$ arrivals and apply greedy nearest-free assignment with uniform tie-breaking.

For Monte Carlo reporting we use the estimator $\widehat{\mu}=\frac{1}{R}\sum_{r=1}^R X_r$, where $X_r$ is the
total cost in run $r$. We report a $95\%$ confidence interval (CI) as
$\widehat{\mu}\pm 1.96\,\widehat{\sigma}/\sqrt{R}$, where $\widehat{\sigma}$ is the sample standard deviation of
$\{X_r\}_{r=1}^R$. Randomness is generated by a standard pseudorandom generator.; Ffor reproducibility we recommend
fixing a base seed and logging it alongside $(n,R)$.

\subsection{Results}\label{sec:exp_results}
Table~\ref{tab:small_n} summarizes expected total cost estimates for small polygons (using $R=10{,}000$ Monte Carlo
runs as a lightweight check). For $n=4$, the estimate matches the closed-form value
$71/32$ from Section~\ref{sec:square}, validating the distance convention and assignment logic.
Table~\ref{tab:large_n} provides a compact reporting format for larger $n$ using direct simulation with
confidence intervals.

\vspace{-0.2em}
\begin{table}[t]
\centering
\caption{Small-$n$ expected total cost (Monte Carlo, $R=10{,}000$ runs). For $n=4$, the exact value is $71/32$.}
\label{tab:small_n}
\vspace{-0.4em}
\begin{tabular}{r l r r}
\toprule
$n$ & Polygon & $\widehat{\mu}$ (total) & $\widehat{\mu}/n$ \\
\midrule
3 & Triangle & 1.414 & 0.471 \\
4 & Square   & 2.222 & 0.556 \\
5 & Pentagon & 3.138 & 0.628 \\
6 & Hexagon  & 4.159 & 0.693 \\
7 & Heptagon & 5.284 & 0.755 \\
8 & Octagon  & 6.493 & 0.812 \\
9 & Nonagon  & 7.783 & 0.865 \\
\bottomrule
\end{tabular}
\end{table}

\begin{table}[t]
\centering
\caption{Direct Simulation results on Large-$n$ values.}
\label{tab:large_n}

\begin{tabular}{r r r r}
\toprule
$n$ & $R$ & $\widehat{\mu}$ (total) & $95\%$ CI for $\widehat{\mu}$ \\
\midrule
20  & 20{,}000 & 22.72 & $[22.67,\;22.77]$ \\
50  & 20{,}000 & 73.20 & $[73.08,\;73.32]$ \\
100 & 20{,}000 & 171.22 & $[171.02,\;171.42]$ \\
\bottomrule
\end{tabular}
\end{table}

\subsection{Trend discussion}\label{sec:exp_trends}
Across all tested $n$, the expected \emph{total} assignment cost increases with $n$.
More importantly, the \emph{per-customer} cost $\widehat{\mu}/n$ also increases gradually as $n$ grows, reflecting that late arrivals are increasingly forced to travel to more distant remaining free vertices.
We emphasize that the observed growth is an empirical trend under the uniform-arrival model. We do not claim a strict proportionality law without separate proof.

\section{Conclusion and Future Work}\label{sec:conc}
We analyzed greedy online facility assignment on regular polygons under an explicit random-arrival model.
Our main theorem gives an exact finite-horizon integral recurrence for the expected remaining cost $V(S)$ from any
occupancy state, and we showed that $V(S)$ is invariant under dihedral symmetries, enabling state canonicalization
and symmetry-reduced dynamic programming.
To compute values reliably within a 12-page LNCS presentation, we split computation by regime: deterministic
integration of the recurrence for small $n$, and direct Monte Carlo simulation with confidence intervals for
larger $n$, validated by a closed-form square computation ($n=4$, $71/32$).

Promising extensions include (i) non-regular polygons or general planar graphs, where symmetry breaks,
(ii) non-uniform or correlated arrival distributions, (iii) capacities greater than one per facility, and
(iv) bridging expected-cost analysis with worst-case competitive guarantees in polygonal metrics~\cite{ahmed2020online,malik2025online}.


\bibliographystyle{splncs04}
\bibliography{samplepaper}

\clearpage
\appendix
\section{Appendix}\label{app:proofs}
\subsection{Monte Carlo Dynamic Programming Source Code}
\begin{lstlisting}
#include <iostream>
#include <vector>
#include <random>
#include <chrono>
#include <iomanip>
#include <cmath>
#include <cstdint>
using namespace std;

// ---------- RNG ----------
static uint64_t SEED = 123456789ULL;
static mt19937_64 rng(SEED);

static inline double urand01() {
    static uniform_real_distribution<double> dist(0.0, 1.0);
    return dist(rng);
}
static inline int irand(int lo, int hi) { // inclusive
    uniform_int_distribution<int> dist(lo, hi);
    return dist(rng);
}

// ---------- Geometry ----------
static inline int delta_cycle(int a, int b, int n) {
    int d = abs(a - b);
    return min(d, n - d);
}

static inline double edgewalk_dist(int n, int e, double t, int v) {
    int i = e;
    int j = (e + 1) % n;       // edge is (i, j)
    double via_i = t + delta_cycle(i, v, n);         // t is distance to i
    double via_j = (1.0 - t) + delta_cycle(j, v, n); // 1-t is distance to j
    return min(via_i, via_j);
}

int main() {
    int n = 9;
    int NUM_SAMPLES = 5000;
    const double EPS = 1e-12;

    cout << fixed << setprecision(6);
    cout << "n=" << n << ", NUM_SAMPLES=" << NUM_SAMPLES << ", SEED=" << SEED << "\n";

    if (n <= 0 || n > 25) {
        cerr << "Warning: state-space DP uses 2^n states; pick small n.\n";
    }

    uint64_t N = 1ULL << n;
    vector<double> V(N, 0.0);

    // bucket masks by popcount
    vector<vector<uint64_t>> by_pc(n + 1);
    by_pc.reserve(n + 1);
    for (uint64_t m = 0; m < N; m++) {
        int pc = __builtin_popcountll(m);
        by_pc[pc].push_back(m);
    }

    // Backward induction: V[full]=0 already
    for (int pc = n - 1; pc >= 0; pc--) {
        for (uint64_t mask : by_pc[pc]) {

            // build free vertices list
            vector<int> freev;
            freev.reserve(n - pc);
            for (int v = 0; v < n; v++) {
                if ((mask & (1ULL << v)) == 0) freev.push_back(v);
            }

            double acc = 0.0;

            for (int s = 0; s < NUM_SAMPLES; s++) {
                int e = irand(0, n - 1);
                double t = urand01();

                double best = 1e100;
                vector<int> ties;
                ties.reserve(4);

                for (int fv : freev) {
                    double d = edgewalk_dist(n, e, t, fv);
                    if (d + EPS < best) {
                        best = d;
                        ties.clear();
                        ties.push_back(fv);
                    } else if (fabs(d - best) <= EPS) {
                        ties.push_back(fv);
                    }
                }

                // uniform tie-break among minimizers
                int chosen = ties[irand(0, (int)ties.size() - 1)];
                uint64_t next_mask = mask | (1ULL << chosen);

                acc += best + V[next_mask];
            }

            V[mask] = acc / (double)NUM_SAMPLES;
        }
    }

    cout << "Estimated V(empty) = " << V[0] << "\n";
    return 0;
}

\end{lstlisting}

\end{document}